\documentclass{llncs}

\usepackage{amsmath,amssymb}
\usepackage{thm-restate}
\usepackage{xfrac}
\usepackage{stmaryrd,wasysym}
\usepackage{tikz}
\usetikzlibrary{automata,arrows,decorations.pathreplacing,petri,shapes.geometric}
\usepackage{todonotes}
\usepackage{hyperref}
\usepackage[inline]{enumitem}
\usepackage[labelformat=parens]{subcaption}
\usepackage{wrapfig}

\input{macros.tex}

\title{The Inconsistent Labelling Problem of Stutter-Preserving Partial-Order Reduction}
\author{Thomas Neele\inst{1} \and Antti Valmari\inst{2} \and Tim A.C. Willemse\inst{1}}
\institute{Eindhoven University of Technology, The Netherlands \and
	University of Jyv\"askyl\"a, Finland}

\begin{document}
\pagestyle{plain}
\maketitle

\begin{abstract}
	In model checking, partial-order reduction (POR) is an effective technique to reduce the size of the state space.
	Stubborn sets are an established variant of POR and have seen many applications over the past 31 years.
	One of the early works on stubborn sets shows that a combination of several conditions on the reduction is sufficient to preserve stutter-trace equivalence, making stubborn sets suitable for model checking of linear-time properties.
	In this paper, we identify a flaw in the reasoning and show with a counter-example that stutter-trace equivalence is not necessarily preserved.
	We propose a solution together with an updated correctness proof.
	Furthermore, we analyse in which formalisms this problem may occur.
	The impact on practical implementations is limited, since they all compute a correct approximation of the theory.
\end{abstract}

\section{Introduction}
In formal methods, model checking is a technique to automatically decide the correctness of a system's design.
The many interleavings of concurrent processes can cause the state space to grow exponentially with the number of components, known as the \emph{state-space explosion} problem.
\emph{Partial-order reduction} (POR) is one technique that can alleviate this problem.
Several variants of POR exist, such as \emph{ample sets}~\cite{Peled1993}, \emph{persistent set}~\cite{Godefroid1996} and \emph{stubborn sets}~\cite{Valmari1991a,Valmari2017a}.
For each of those variants, sufficient conditions for preservation of stutter-trace equivalence have been identified.
Since LTL without the next operator (LTL$_{-X}$) is invariant under finite stuttering, this allows one to check most LTL properties under POR.

However, the correctness proofs for these methods are intricate and not reproduced often.
For stubborn sets, LTL$_{-X}$-preserving conditions and an accompanying correctness result were first presented in~\cite{Valmari1991b}, and discussed in more detail in~\cite{Valmari1992}.
While trying to reproduce the proof for~\cite[Theorem 2]{Valmari1992} (see also Theorem~\ref{thm:d1_preserve_stutter_trace_equivalence} in the current work), we ran into an issue while trying to prove a certain property of the construction used in the original proof~\cite[Construction 1]{Valmari1992}.
This led us to discover that stutter-trace equivalence is not necessarily preserved.
We will refer to this as the \emph{inconsistent labelling problem}.
The essence of the problem is that POR in general, and the proofs in~\cite{Valmari1992} in particular, reason mostly about actions, which label the transitions.
The only relevance of the state labelling is that it determines which actions are \emph{visible}.
On the other hand, stutter-trace equivalence and the LTL semantics are purely based on state labels.
The correctness proof in~\cite{Valmari1992} does not deal properly with this disparity.
Further investigation shows that the same problem also occurs in two works of Bene\v{s} \etal~\cite{Benes2011,Benes2009}, who apply ample sets to state/event LTL model checking.

Consequently, any application of stubborn sets in LTL$_{-X}$ model checking is possibly unsound, both for safety and liveness properties.
In literature, the correctness of several theories~\cite{Laarman2016,Liebke2019,Valmari1996} relies on the incorrect theorem.

Our contributions are as follows:
\begin{itemize}
	\item We prove the existence of the inconsistent labelling problem with a counter-example.
	This counter-example is valid for weak stubborn sets and, with a small modification, in a non-deterministic setting for strong stubborn sets.
	\item We propose to strengthen one of the stubborn set conditions and show that this modification resolves the issue (Theorem~\ref{thm:correctness_d1'}).
	\item We analyse in which circumstances the inconsistent labelling problem occurs and, based on the conclusions, discuss its impact on existing literature.
	This includes a thorough analysis of Petri nets and several different notions of invisible transitions and atomic propositions.
\end{itemize}
Our investigation shows that probably all practical implementations of stubborn sets compute an approximation which resolves the inconsistent labelling problem.
Furthermore, POR methods based on the standard independence relation, such as ample sets and persistent sets, are not affected.

The rest of the paper is structured as follows.
In Section~\ref{sec:preliminaries}, we introduce the basic concepts of stubborn sets and stutter-trace equivalence, which is not preserved in the counter-example of Section~\ref{sec:counter_example}.
A solution to the inconsistent labelling problem is discussed in Section~\ref{sec:strengthen_d1}, together with an updated correctness proof.
Sections~\ref{sec:safe_formalisms} and~\ref{sec:petri_nets} discuss several settings in which correctness is not affected.
Finally, Section~\ref{sec:related_work} presents related work and Section~\ref{sec:conclusion} presents a conclusion.

\section{Preliminaries}
\label{sec:preliminaries}
Since LTL relies on state labels and POR relies on edge labels, we assume the existence of some fixed set of atomic propositions \AP to label the states and a fixed set of edge labels \Act, which we will call \emph{actions}.
Actions are typically denoted with the letter $\act$.

\begin{definition}
	\label{def:labelled_LSTS}
	A \emph{labelled state transition system}, short \emph{LSTS}, is a directed graph $\TS = (S,\edgerel,\init{s},L)$, where:
	\begin{itemize}
		\item $S$ is the state space;
		\item $\edgerel \, \subseteq S \times \Act \times S$ is the transition relation;
		\item $\init{s} \in S$ is the initial state; and
		\item $L: S \to 2^{\AP}$ is a function that labels states with atomic propositions.
	\end{itemize}
\end{definition}

We write $s \transition{\act} t$ whenever $(s,\act,t) \in\, \edgerel$.
A path is a (finite or infinite) alternating sequence of states and actions: $s_0 \transition{\act_1} s_1 \transition{\act_2} s_2 \dots$.
We sometimes omit the intermediate and/or final states if they are clear from the context or not relevant, and write $s \transition{\act_1\dots \act_n} t$ or $s \transition{\act_1\dots \act_n}$.
Paths that start in the initial state $\init{s}$ are called \emph{initial paths}.
Given a path $\pi = s_0 \transition{\act_1} s_1 \transition{\act_2} s_2 \dots$, the trace of $\pi$ is the sequence of state labels observed along $\pi$, \viz $L(s_0) L(s_1) L(s_2) \dots$.
An action $\act$ is enabled in a state $s$, notation $s \transition{\act}$, if and only if there is a transition $s \transition{\act} t$ for some $t$.
In a given LSTS \TS, $\enabled_{\TS}(s)$ is the set of all enabled actions in a state $s$.
A set $\Inv$ of invisible actions is chosen such that if (but not necessarily only if) $a \in \Inv$, then for all states $s$ and $t$, $s \transition{\act} t$ implies $L(s) = L(t)$.
Note that this definition allows the set $\Inv$ to be under-approximated.
An action that is not invisible is called \emph{visible}.
We say \TS is \emph{deterministic} if and only if $s \transition{a} t$ and $s \transition{a} t'$ imply $t = t'$, for all states $s$, $t$ and $t'$ and actions $\act$.
To indicate that \TS is not necessarily deterministic, we say \TS is \emph{non-deterministic}.

\subsection{Stubborn sets}
In POR, \emph{reduction functions} play a central role.
A reduction function $\redf: S \to 2^\Act$ indicates which transitions to explore in each state.
When starting at the initial state $\init{s}$, a reduction function induces a \emph{reduced LSTS} as follows.

\begin{definition}
	Let $\TS = (S,\edgerel,\init{s},L)$ be an LSTS and $\redf: S \to 2^\Act$ a reduction function.
	Then the \emph{reduced LSTS} induced by $\redf$ is defined as $\TS_r = (S_r,\rededgerel,\init{s},L_r)$, where $L_r$ is the restriction of $L$ on $S_r$ and $S_r$ and $\rededgerel$ are the smallest sets such that the following holds:
	\begin{itemize}
		\item $\init{s} \in S_r$; and
		\item If $s \in S_r$, $s \transition{\act} t$ and $\act \in \redf(s)$, then $t \in S_r$ and $s \redtransition{\act} t$.
	\end{itemize}
\end{definition}

In the remainder of this paper, we will assume the state space of a reduced LSTS is finite.
This is essential for the correctness of the approach detailed below.
In general, a reduction function is not guaranteed to preserve almost any property of an LSTS.
Below, we list a number of conditions that have been proposed in literature; they aim to preserve LTL$_{-X}$.
Here, an action $\act$ that satisfies $s' \transition{\act}$, for all $\act_1,\dots,\act_n \notin \redf(s)$ and paths $s \transition{\act_1\dots\act_n} s'$, is called a \emph{key action} in $s$, typically denoted with $\keyact$.
\begin{description}[leftmargin=!,labelwidth=\widthof{\bfseries D2w}]
	\item[\textbf{D0}] If $\enabled(s) \neq \emptyset$, then $\redf(s) \cap \enabled(s) \neq \emptyset$.
	\item[\textbf{D1}] For all $\act \in \redf(s)$ and $\act_1,\dots,\act_n \notin \redf(s)$, if $s \transition{\act_1} \dots \transition{\act_n} s_n \transition{\act} s'_n$, then there are states $s',s'_1,\dots,s'_{n-1}$ such that $s \transition{\act} s' \transition{\act_1} s'_1 \transition{\act_2} \dots \transition{\act_n} s'_n$.
	\item[\textbf{D2}] For all actions $\act \in \redf(s)$ and $\act_1,\dots,\act_n \notin \redf(s)$, if $s \transition{a}$ and $s \transition{\act_1\dots  \act_n} s'$, then $s' \transition{\act}$.
	\item[\textbf{D2w}] If $\enabled(s) \neq \emptyset$, then there is an action $\act \in \redf(s)$ such that for all $\act_1,\dots,\act_n \notin \redf(s)$, if $s \transition{\act_1\dots  \act_n} s'$, then $s' \transition{\act}$.
	\item[\textbf{V}] If $\redf(s)$ contains an enabled visible action, then it contains all visible actions.
	\item[\textbf{I}] If an invisible action is enabled, then $\redf(s)$ contains an invisible key action.
	\item[\textbf{L}] For every visible action $\act$, every cycle in the reduced LSTS contains a state $s$ such that $\act \in \redf(s)$.
\end{description}

These conditions are used to define \emph{strong} and \emph{weak} stubborn sets in the following way.

\begin{definition}
	A reduction function $r: S \to 2^\Act$ is a \emph{strong stubborn set} iff for all states $s \in S$, the conditions \textbf{D0}, \textbf{D1}, \textbf{D2}, \textbf{V}, \textbf{I}, \textbf{L} all hold.
\end{definition}
\begin{definition}
	A reduction function $r: S \to 2^\Act$ is a \emph{weak stubborn set} iff for all states $s \in S$, the conditions \textbf{D1}, \textbf{D2w}, \textbf{V}, \textbf{I}, \textbf{L} all hold.
\end{definition}

Below, we also use `weak/strong stubborn set' to refer to the set of actions $\redf(s)$ in some state $s$.
First, note that key actions are always enabled, by setting $n = 0$.
Furthermore, a stubborn set can never introduce new deadlocks, either by \textbf{D0} or \textbf{D2w}.
Condition \textbf{D1} enforces that a key action $\keyact \in \redf(s)$ does not disable other paths that are not selected for the stubborn set.
A visual representation of condition \textbf{D1} can be found in Figure~\ref{fig:condition_D1}.
When combined, \textbf{D1} and \textbf{D2w} are sufficient conditions for preservation of deadlocks.
Condition \textbf{V} enforces that the paths $s \transition{\act_1\dots\act_n\act} s'_n$ and $s \transition{\act\act_1\dots\act_n} s'_n$ in \textbf{D1} contain the same sequence of visible actions.
The purpose of condition \textbf{I} is to preserve the possibility to perform an invisible action, if one is enabled.
Finally, we have condition \textbf{L} to deal with the \emph{action-ignoring problem}, which occurs when an action is never selected for the stubborn set and always ignored.
Since we assume that the reduced LSTS is finite, it suffices to reason in \textbf{L} about every cycle instead of every infinite path.
The combination of \textbf{I} and \textbf{L} helps to preserve divergences.

Conditions \textbf{D0} and \textbf{D2} together imply \textbf{D2w}, and thus every strong stubborn set is also a weak stubborn set.
Since the reverse does not necessarily hold, weak stubborn sets might offer more reduction.

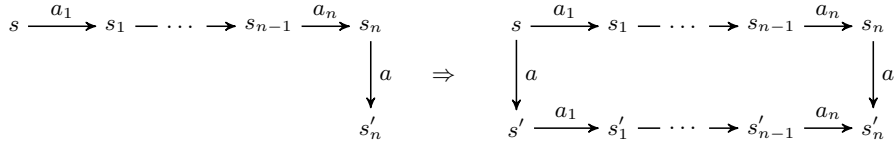
\begin{figure}[t]
	\centering
	\resizebox{\textwidth}{!}{
		\begin{tikzpicture}[->,>=stealth',shorten >=0pt,auto,node distance=2.0cm,semithick]
		\begin{scope}
			\def\d{1.4}
			\node (s)     at (0,\d)        {$s$};
			\node (s1)    at (\d,\d)       {$s_1$};
			\node (d)     at (1.65*\d,\d)  {$\dots$};
			\node (sn-1)  at (2.5*\d,\d)   {$s_{n-1}$};
			\node (sn)    at (3.5*\d,\d)   {$s_n$};
			\node (s'n)   at (3.5*\d,0)    {$s'_n$};
			\path
			(s)     edge node {$\act_1$} (s1)
			(d)     edge              (sn-1)
			(sn-1)  edge node {$\act_n$} (sn)
			(sn)    edge node {$\act$}   (s'n);
			\path[-]
			(s1)    edge              (d);
		\end{scope}
		\node at (5.9,0.7) {$\Rightarrow$};
		\begin{scope}[xshift=6.9cm]
			\def\d{1.4}
			\node (s)     at (0,\d)        {$s$};
			\node (s1)    at (\d,\d)       {$s_1$};
			\node (d)     at (1.65*\d,\d)  {$\dots$};
			\node (sn-1)  at (2.5*\d,\d)   {$s_{n-1}$};
			\node (sn)    at (3.5*\d,\d)   {$s_n$};
			\node (s')    at (0,0)         {$s'$};
			\node (s'1)   at (\d,0)        {$s'_1$};
			\node (d')    at (1.65*\d,0)   {$\dots$};
			\node (s'n-1) at (2.5*\d,0)    {$s'_{n-1}$};
			\node (s'n)   at (3.5*\d,0)    {$s'_n$};
			\path
			(s)     edge node {$\act_1$} (s1)
			(d)     edge              (sn-1)
			(sn-1)  edge node {$\act_n$} (sn)
			(sn)    edge node {$\act$}   (s'n);
			\path[-]
			(s1)    edge              (d);
			\path
			(s')    edge node {$\act_1$} (s'1)
			(d')    edge              (s'n-1)
			(s'n-1) edge node {$\act_n$} (s'n)
			(s)     edge node {$\act$}   (s');
			\path[-]
			(s'1)   edge              (d');
		\end{scope}
		\end{tikzpicture}
	}
	\caption{Visual representation of condition \textbf{D1}.}
	\label{fig:condition_D1}
\end{figure}

\subsection{Weak and Stutter Equivalence}
To reason about the similarity of an LSTS $\TS$ and its reduced LSTS $\TS_r$, we introduce the notions of \emph{weak equivalence}, which operates on actions, and \emph{stutter equivalence}, which operates on states.
The definitions are generic, so that they can also be used in Section~\ref{sec:petri_nets}.
\begin{definition}
	The paths $\pi$ and $\pi'$ are weakly equivalent with respect to a set of actions $I$, notation $\pi \weakeq_I \pi'$, if and only if they are both finite or both infinite and their respective projections on $\Act \setminus I$ are equal.
\end{definition}
\begin{definition}
	The \emph{no-stutter trace} under labelling $L$ of a path $s_0 \transition{\act_1} s_1 \transition{\act_2} \dots$ is the sequence of those $L(s_i)$ such that $i = 0$ or $L(s_i) \neq L(s_{i-1})$.
	Paths $\pi$ and $\pi'$ are stutter equivalent under $L$, notation $\pi \stuteq_L \pi'$, iff they are both finite or both infinite, and they yield the same no-stutter trace under $L$.
\end{definition}

We typically consider weak equivalence with respect to the set of invisible actions $\Inv$.
In that case, we write $\pi \weakeq \pi'$.
We also omit the subscript for stutter equivalence when reasoning about the standard labelling function and write $\pi \stuteq \pi'$.
Remark that stutter equivalence is invariant under finite repetitions of state labels, hence its name.
We lift both equivalences to LSTSs, and say that $\TS$ and $\TS'$, having the same set of invisible actions, are \emph{weak-trace equivalent}  iff for every initial path $\pi$ in $\TS$, there is a weakly equivalent initial path $\pi'$ in $\TS'$ and vice versa.
Likewise, $\TS$ and $\TS'$, having the same labelling function, are \emph{stutter-trace equivalent} iff for every initial path $\pi$ in $\TS$, there is a stutter equivalent initial path $\pi'$ in $\TS'$ and vice versa.
Since we only reason about equivalence between a full LSTS $\TS$ and its reduced LSTS $\TS_r$, which is a subgraph of $\TS$, the assumption that both have the same set of invisible actions and the same state labelling is not limiting.

In general, weak equivalence and stutter equivalence are incomparable, even for initial paths.
However, for some LSTSs, these notions can be related in a certain way.
We formalise this in the following definition.
\begin{definition}
	Let $\TS$ be an LSTS and $\pi$ and $\pi'$ two paths in \TS that both start in some state $s$.
	Then, $\TS$ is \emph{labelled consistently} iff $\pi \weakeq \pi'$ implies $\pi \stuteq \pi'$.
\end{definition}

Note that if an LSTS is labelled consistently, then in particular all weakly equivalent initial paths are also stutter equivalent.
Hence, if an LSTS $\TS$ is labelled consistently and weak-trace equivalent to a subgraph $\TS'$, then $\TS$ and $\TS'$ are also stutter-trace equivalent.

Stubborn sets as defined in the previous section aim to preserve stutter-trace equivalence between the original and the reduced LSTS.
The motivation behind this is that two stutter-trace equivalent LSTSs satisfy exactly the same formulae~\cite{BaierKatoen-PMC} in LTL$_{-X}$ (LTL without the next operator).
The following theorem, which is frequently cited in literature~\cite{Laarman2016,Liebke2019,Valmari1996}, aims to show that stubborn sets indeed preserve stutter-trace equivalence.
Its original formulation reasons about the validity of an arbitrary LTL$_{-X}$ formula.
Here, we give the alternative formulation based on stutter-trace equivalence.

\begin{theorem}{\cite[Theorem 2]{Valmari1992}}
	\label{thm:d1_preserve_stutter_trace_equivalence}
	Given an LSTS $\TS$ and a weak/strong stubborn set $\redf$, then the reduced LSTS $\TS_\redf$ is stutter-trace equivalent to $\TS$.
\end{theorem}

The original proof correctly concludes that the stubborn set method preserves the order of visible actions in the reduced LSTS, \ie, $\TS \weakeq \TS_r$.
However, this only implies preservation of stutter-trace equivalence ($\TS \stuteq \TS_r$) if the full LSTS is labelled consistently, so Theorem~\ref{thm:d1_preserve_stutter_trace_equivalence} is invalid in the general case.
In the next section, we will see a counter-example which exploits this fact.

\section{Counter-Example}
\label{sec:counter_example}
Consider the LSTS in Figure~\ref{fig:counter_example}, which we will refer to with $\TS^C$.
There is only one atomic proposition $q$, which holds in the grey states and is false in the other states.
The initial state $\init{s}$ is marked with an incoming arrow.
First, note that this LSTS is deterministic.
The actions $\act_1$, $\act_2$ and $\act_3$ are visible and $\act$ and $\keyact$ are invisible.
By setting $r(\init{s}) = \{\act,\keyact\}$, which is a weak stubborn set, we obtain a reduced LSTS $\TS^C_r$ that does not contain the dashed states and transitions.
The original LSTS contains the trace $\emptyset \{q\} \emptyset \emptyset \{q\}^\omega$, obtained by following the path with actions $a_1 a_2 a a_3^\omega$.
However, the reduced LSTS does not contain a stutter equivalent trace.
This is also witnessed by the LTL$_{-X}$ formula $\square (q \Rightarrow \square(q \lor \square \neg q))$, which holds for $\TS^C_r$, but not for $\TS^C$.

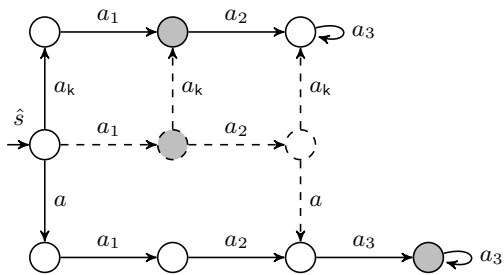
\begin{figure}
	\centering
	\begin{tikzpicture}[->,>=stealth',shorten >=0pt,auto,node distance=2.0cm,semithick]
	\tikzstyle{state}=[draw,inner sep=4pt,circle]
	\def\x{1.7}
	\def\y{1.5}
	
	\node[state,label={above left:$\init{s}$}] (1)  at (0,\y)      {};
	\node[state,dashed,fill=lightgray]         (2)  at (\x,\y)     {};
	\node[state,dashed]                        (3)  at (2*\x,\y)   {};
	\node[state]                               (4)  at (0,0)       {};
	\node[state]                               (5)  at (\x,0)      {};
	\node[state]                               (6)  at (2*\x,0)    {};
	\node[state]                               (7)  at (0,2*\y)    {};
	\node[state,fill=lightgray]                (8)  at (\x,2*\y)   {};
	\node[state]                               (9)  at (2*\x,2*\y) {};
	\node[state,fill=lightgray]                (10) at (3*\x,0)    {};
	\path
		(-0.5,\y) edge (1)
		(1)  edge             node {$\act$}   (4)
		(4)  edge             node {$\act_1$} (5)
		(5)  edge             node {$\act_2$} (6)
		(1)  edge[']          node {$\keyact$} (7)
		(7)  edge             node {$\act_1$} (8)
		(8)  edge             node {$\act_2$} (9)
		(9)  edge[loop right] node {$\act_3$} (9)
		(6)  edge             node {$\act_3$} (10)
		(10) edge[loop right] node {$\act_3$} (10)
	;
	\path[dashed]
		(1)  edge             node {$\act_1$} (2)
		(2)  edge             node {$\act_2$} (3)
		(3)  edge             node {$\act$}   (6)
		(3)  edge[']          node {$\keyact$} (9)
		(2)  edge[']          node {$\keyact$} (8)
	;
	\end{tikzpicture}
	\caption{Counter-example showing that stubborn sets do not preserve stutter-trace equivalence.
	Grey state are labelled with $\{ q \}$.
	The dashed transitions and states are not present in the reduced LSTS.}
	\label{fig:counter_example}
\end{figure}

A very similar example can be used to show that strong stubborn sets suffer from the same problem.
Consider again the LSTS in Figure~\ref{fig:counter_example}, but assume that $a = \keyact$, making the LSTS non-deterministic.
Now, $r(\init{s}) = \{\act\}$ is a strong stubborn set and again the trace $\emptyset \{q\} \emptyset \emptyset \{q\}^\omega$ is not preserved in the reduced LSTS.
In Section~\ref{sec:deterministic_lstss}, we will see why the inconsistent labelling problem does not occur for deterministic systems under strong stubborn sets.

The core of the problem lies in the fact that condition \textbf{D1}, even when combined with \textbf{V}, does not enforce that the two paths it considers are stutter equivalent.
Consider the paths $s \transition{\act}$ and $s \transition{\act_1 \act_2 \act}$ and assume that $\act \in \redf(s)$ and $\act_1, \act_2 \notin \redf(s)$.
Conditions \textbf{D1} and \textbf{V} ensure that at least one of the following two holds:
\begin{enumerate*}[label=\textnormal{(\roman*)}]
	\item $\act$ is invisible, or
	\item $\act_1$ and $\act_2$ are invisible
\end{enumerate*}.
Half of the possible scenarios are depicted in Figure~\ref{fig:paths_d1}; the other half are symmetric.
Again, the grey states (and only those states) are labelled with $\{ q \}$.

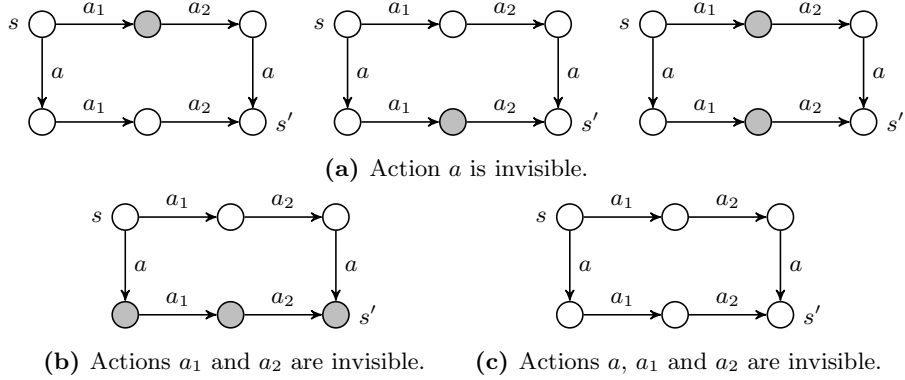
\begin{figure}
	\centering
	\begin{subfigure}{\textwidth}
		\centering
		\begin{tikzpicture}[->,>=stealth',shorten >=0pt,auto,node distance=2.0cm,semithick]
		\tikzstyle{state}=[draw,inner sep=3.5pt,circle]
		\def\x{1.4}
		\def\y{1.3}
		
		\begin{scope}
		\node[state,label={left:$s$}]   (1)  at (0,\y)      {};
		\node[state,fill=lightgray]     (2)  at (\x,\y)     {};
		\node[state]                    (3)  at (2*\x,\y)   {};
		\node[state]                    (4)  at (0,0)       {};
		\node[state]                    (5)  at (\x,0)      {};
		\node[state,label={right:$s'$}] (6)  at (2*\x,0)    {};
		\path
		(1)  edge node {$\act$}   (4)
		(4)  edge node {$\act_1$} (5)
		(5)  edge node {$\act_2$} (6)
		(1)  edge node {$\act_1$} (2)
		(2)  edge node {$\act_2$} (3)
		(3)  edge node {$\act$}   (6)
		;
		\end{scope}
		
		\begin{scope}[xshift=2.9*\x cm]
		\node[state,label={left:$s$}]   (1)  at (0,\y)      {};
		\node[state]                    (2)  at (\x,\y)     {};
		\node[state]                    (3)  at (2*\x,\y)   {};
		\node[state]                    (4)  at (0,0)       {};
		\node[state,fill=lightgray]     (5)  at (\x,0)      {};
		\node[state,label={right:$s'$}] (6)  at (2*\x,0)    {};
		\path
		(1)  edge node {$\act$}   (4)
		(4)  edge node {$\act_1$} (5)
		(5)  edge node {$\act_2$} (6)
		(1)  edge node {$\act_1$} (2)
		(2)  edge node {$\act_2$} (3)
		(3)  edge node {$\act$}   (6)
		;
		\end{scope}
		
		\begin{scope}[xshift=5.8*\x cm]
		\node[state,label={left:$s$}]   (1)  at (0,\y)      {};
		\node[state,fill=lightgray]     (2)  at (\x,\y)     {};
		\node[state]                    (3)  at (2*\x,\y)   {};
		\node[state]                    (4)  at (0,0)       {};
		\node[state,fill=lightgray]     (5)  at (\x,0)      {};
		\node[state,label={right:$s'$}] (6)  at (2*\x,0)    {};
		\path
		(1)  edge node {$\act$}   (4)
		(4)  edge node {$\act_1$} (5)
		(5)  edge node {$\act_2$} (6)
		(1)  edge node {$\act_1$} (2)
		(2)  edge node {$\act_2$} (3)
		(3)  edge node {$\act$}   (6)
		;
		\end{scope}
		\end{tikzpicture}
		\caption{Action $\act$ is invisible.}
		\label{fig:paths_d1_a_invisible}
	\end{subfigure}
	\begin{subfigure}{0.45\textwidth}
		\centering
		\begin{tikzpicture}[->,>=stealth',shorten >=0pt,auto,node distance=2.0cm,semithick]
		\tikzstyle{state}=[draw,inner sep=3.5pt,circle]
		\def\x{1.4}
		\def\y{1.3}
		
		\begin{scope}
		\node[state,label={left:$s$}]                  (1)  at (0,\y)      {};
		\node[state]                                   (2)  at (\x,\y)     {};
		\node[state]                                   (3)  at (2*\x,\y)   {};
		\node[state,fill=lightgray]                    (4)  at (0,0)       {};
		\node[state,fill=lightgray]                    (5)  at (\x,0)      {};
		\node[state,fill=lightgray,label={right:$s'$}] (6)  at (2*\x,0)    {};
		\path
		(1)  edge node {$\act$}   (4)
		(4)  edge node {$\act_1$} (5)
		(5)  edge node {$\act_2$} (6)
		(1)  edge node {$\act_1$} (2)
		(2)  edge node {$\act_2$} (3)
		(3)  edge node {$\act$}   (6)
		;
		\end{scope}
		\end{tikzpicture}
		\caption{Actions $\act_1$ and $\act_2$ are invisible.}
		\label{fig:paths_d1_a12_invisible}
	\end{subfigure}
	\hspace{0.2cm}
	\begin{subfigure}{0.45\textwidth}
		\centering
		\begin{tikzpicture}[->,>=stealth',shorten >=0pt,auto,node distance=2.0cm,semithick]
		\tikzstyle{state}=[draw,inner sep=3.5pt,circle]
		\def\x{1.4}
		\def\y{1.3}
		
		\node[state,label={left:$s$}]    (1)  at (0,\y)      {};
		\node[state]                     (2)  at (\x,\y)     {};
		\node[state]                     (3)  at (2*\x,\y)   {};
		\node[state]                     (4)  at (0,0)       {};
		\node[state]                     (5)  at (\x,0)      {};
		\node[state,label={right:$s'$}]  (6)  at (2*\x,0)    {};
		\path
		(1)  edge node {$\act$}   (4)
		(4)  edge node {$\act_1$} (5)
		(5)  edge node {$\act_2$} (6)
		(1)  edge node {$\act_1$} (2)
		(2)  edge node {$\act_2$} (3)
		(3)  edge node {$\act$}   (6)
		;
		\end{tikzpicture}
		\caption{Actions $\act$, $\act_1$ and $\act_2$ are invisible.}
		\label{fig:paths_d1_aa12_invisible}
	\end{subfigure}
	\caption{Five possible scenarios when $\act \in \redf(s)$ and $\act_1, \act_2 \notin \redf(s)$, according to conditions \textbf{D1} and \textbf{V}.}
	\label{fig:paths_d1}
\end{figure}

The left and middle cases in Figure~\ref{fig:paths_d1_a_invisible} are problematic.
In both LSTSs, the paths $s \transition{\act_1 \act_2 \act} s'$ and $s \transition{\act \act_1 \act_2} s'$ are weakly equivalent, since $a$ is invisible.
However, they are not stutter equivalent, and therefore these LSTSs are not labelled consistently.
The left LSTS in Figure~\ref{fig:paths_d1_a_invisible} forms the core of the counter-example $\TS^C$, with the rest of $\TS^C$ serving to satisfy condition \textbf{D2}/\textbf{D2w}.

\section{Strengthening Condition D1}
\label{sec:strengthen_d1}
To fix the issue with inconsistent labelling, we propose the following alternative, stronger version of \textbf{D1}.
\begin{description}[leftmargin=!,labelwidth=\widthof{\bfseries D1'}]
	\item[\textbf{D1'}] For all $\act \in \redf(s)$ and $\act_1,\dots,\act_n \notin \redf(s)$, if $s \transition{\act_1} s_1 \transition{\act_2} \dots \transition{\act_n} s_n \transition{\act} s'_n$, then there are states $s',s'_1,\dots,s'_{n-1}$ such that $s \transition{\act} s' \transition{\act_1} s'_1 \transition{\act_2} \dots \transition{\act_n} s'_n$.
	Furthermore, if $\act$ is invisible, then $s_i \transition{\act} s'_i$ for every $1 \leq i < n$.
\end{description}

This new condition \textbf{D1'} provides a form of \emph{local} consistent labelling when one of $\act_1,\dots,\act_n$ is visible.
In this case, \textbf{V} implies that $\act$ is invisible and, consequently, the presence of transitions $s_i \transition{\act} s'_i$ implies $L(s_i) = L(s'_i)$.
Hence, the problematic cases of Figure~\ref{fig:paths_d1_a_invisible} are resolved; a correctness proof is given below.

Condition \textbf{D1'} is very similar to condition \textbf{C1}~\cite{Gerth1999}, which is common in the context of ample sets.
However, \textbf{C1} requires that action $\act$ is \emph{globally} independent of each of the actions $\act_1,\dots,\act_n$, while \textbf{D1'} merely requires a kind of \emph{local} independence.
In practice, though, most, if not all, implementations compute a global independence relation.
Persistent sets~\cite{Godefroid1996} also rely on a condition similar to \textbf{D1'}, and require local independence.

\subsection{Correctness}
\label{sec:correctness}
To show that \textbf{D1'} indeed resolves the inconsistent labelling problem, we reproduce the construction in the original proof~\cite[Construction 1]{Valmari1992} in two lemmata and show that it preserves stutter equivalence.
Below, we often use $\rededgerel$ to indicate which transitions must occur in the reduced state space.
\begin{lemma}
	\label{lmm:shift_action_forward}
	Let $\redf$ be a weak stubborn set and $\pi = s_0 \transition{\act_1} \dots \transition{\act_n} s_n \transition{\act} s'_n$ be a path such that $\act_1,\dots,\act_n \notin \redf(s_0)$ and $\act \in \redf(s_0)$.
	Then, there is a path $\pi' = s_0 \redtransition{\act} s'_0 \transition{\act_1} \dots \transition{\act_n} s'_n$ such that $\pi$ and $\pi'$ are stutter equivalent.
\end{lemma}
\begin{proof}
	The existence of $\pi'$ follows directly from condition \textbf{D1'}.
	Due to condition \textbf{V} and our assumption that $\act_1,\dots,\act_n \notin \redf(s_0)$, it cannot be the case that $\act$ is visible and at least one of $\act_1,\dots,\act_n$ is visible.
	If $\act$ is invisible, then the traces of $s_0 \transition{\act_1} \dots \transition{\act_n} s_n$ and $s'_0 \transition{\act_1} \dots \transition{\act_n} s'_n$ are equivalent, since \textbf{D1'} implies that $s_i \transition{\act} s'_i$ for every $0 \leq i \leq n$, so $L(s'_i) = L(s_i)$.
	Otherwise, if all of $\act_1,\dots,\act_n$ are invisible, then the sequences of labels observed along $\pi$ and $\pi'$ have the shape $L(s_0)^{n+1} L(s'_0)$ and $L(s_0) L(s'_0)^{n+1}$, respectively.
	We conclude that $\pi$ and $\pi'$ are stutter equivalent.
	\qed
\end{proof}
\begin{lemma}
	\label{lmm:introduce_key_action}
	Let $\redf$ be a weak stubborn set and $\pi = s_0 \transition{\act_1} s_1 \transition{\act_2} \dots$ be a path such that $\act_i \notin \redf(s_0)$ for any $\act_i$ that occurs in $\pi$.
	Then, the following holds:
	\begin{itemize}
		\item If $\pi$ is of finite length $n > 0$, there exist an action $\keyact$, a state $s'_n$ such that $s_n \transition{\keyact} s'_n$ and a path $\pi' = s_0 \redtransition{\keyact} s'_0 \transition{\act_1} \dots \transition{\act_n} s'_n$.
		\item If $\pi$ is infinite, there exists a path $\pi' = s_0 \redtransition{\keyact} s'_0 \transition{\act_1} s'_1 \transition{\act_2} \dots$ for some action $\keyact$.
	\end{itemize}
	In either case, $\pi \stuteq \pi'$.
\end{lemma}
\begin{proof}
	Let $K$ be the set of key actions in $s$.
	If $\act_1$ is invisible, $K$ contains at least one invisible action, due to \textbf{I}.
	Otherwise, if $\act_1$ is visible, we reason that $K$ is not empty (condition \textbf{D2w}) and all actions in $\redf(s_0)$, and thus also all actions in $K$, are invisible, due to \textbf{V}.
	In the remainder, let $\keyact$ be an invisible key action.	
	
	In case $\pi$ has finite length $n$, the existence of $s_n \transition{\keyact} s'_n$ and $s_0 \redtransition{\keyact} s'_0 \transition{\act_1} \dots \transition{\act_n} s'_n$ follows from \textbf{D2w} and \textbf{D1'}, respectively.
	
	If $\pi$ is infinite, we can apply \textbf{D2w} and \textbf{D1'} successively to obtain a path $\pi_i = s_0 \transition{\keyact} s'_0 \transition{\act_1} \dots \transition{\act_i} s'_i$ for every $i \geq 0$.
	Since the reduced state space is finite, infinitely many of these paths must use the same state as $s'_0$.
	At most one of them ends at $s'_0$ (the one with $i = 0$), so infinitely many continue from $s'_0$.
	Of them, infinitely many must use the same $s'_1$, again because the reduced state space is finite.
	Again, at most one of them is lost because of ending at $s'_1$.
	This reasoning can continue without limit, proving the existence of $\pi' = s_0 \redtransition{\keyact} s'_0 \transition{\act_1} s'_1 \transition{\act_2} \dots$.
	
	Since $\keyact$ is invisible, we use the same reasoning as in the proof of Lemma~\ref{lmm:shift_action_forward} to conclude $\pi \stuteq \pi'$.
	\qed
\end{proof}

Lemmata~\ref{lmm:shift_action_forward} and~\ref{lmm:introduce_key_action} coincide with branches 1 and 2 of~\cite[Construction 1]{Valmari1992}, respectively, but contain the stronger result that $\pi \stuteq \pi'$.
Thus, when applied in the proof of~\cite[Theorem 2]{Valmari1992} (see also Theorem~\ref{thm:d1_preserve_stutter_trace_equivalence}), this yields the result that stubborn sets with condition \textbf{D1'} preserve stutter-trace equivalence.

\begin{theorem}
	\label{thm:correctness_d1'}
	Given an LSTS $\TS$ and weak/strong stubborn set $r$, where condition \textbf{D1} is replaced by \textbf{D1'}, then the reduced LSTS $\TS_r$ is stutter-trace equivalent to $\TS$.
\end{theorem}

We do not reproduce the complete proof, but provide insight into the application of the lemmata with the following example.

\begin{example}
	Consider the path obtained by following $\act_1 \act_2 \act_3$ in Figure~\ref{fig:example_mimicking}.
	Lemmata~\ref{lmm:shift_action_forward} and~\ref{lmm:introduce_key_action} show that $\act_1 \act_2 \act_3$ can always be mimicked in the reduced LSTS, while preserving stutter equivalence.
	In this case, the path is mimicked by the path corresponding to $\keyact \act_2 \act_1 \keyact' \act_3$, drawn with dashes.
	The new path reorders the actions $\act_1$, $\act_2$ and $\act_3$ according to the construction of Lemma~\ref{lmm:shift_action_forward} and introduces the key actions $\keyact$ and $\keyact'$ according to the construction of Lemma~\ref{lmm:introduce_key_action}.
	\qed
	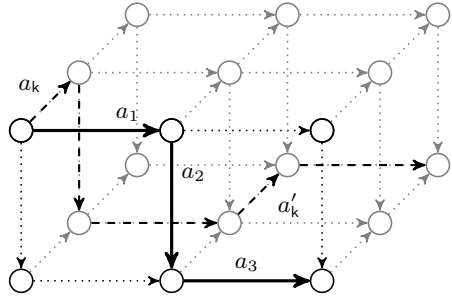
\begin{figure}
		\centering
		\begin{tikzpicture}[->,>=stealth',shorten >=0pt,auto,node distance=2.0cm,semithick]
		\tikzstyle{state} = [draw,circle]
		\def\d{2}
		\def\dx{\d}
		\def\dy{-\d}
		\def\dz{-\d}
		\node[state] (000) at (0,0,0) {};
		\node[state] (100) at (\d,0,0) {};
		\node[state] (200) at (2*\d,0,0) {};
		\foreach \x in {0,1,2} {
			\foreach \y in {0,1} {
				\node[state] (\x\y0) at (\x*\dx,\y*\dy,0) {};
			}
		}
		\foreach \x in {0,1,2} {
			\foreach \y in {0,1} {
				\foreach \z in {1,2} {
					\node[state,gray] (\x\y\z) at (\x*\dx,\y*\dy,\z*\dz) {};
				}
			}
		}
		\foreach \y in {0,1} {
			\path[dotted] (0\y0) edge (1\y0)
			(1\y0) edge (2\y0);
		}
		\foreach \x in {0,1,2} {
			\path[dotted] (\x00) edge (\x10);
		}theorem
		\foreach \x in {0,1,2} {
			\foreach \y in {0,1} {
				\path[dotted,gray] (\x\y0) edge (\x\y1)
				(\x\y1) edge (\x\y2);
			}
		}
		\foreach \y in {0,1} {
			\foreach \z in {1,2} {
				\path[dotted,gray] (0\y\z) edge (1\y\z)
				(1\y\z) edge (2\y\z);
			}
		}
		\foreach \x in {0,1,2} {
			\foreach \z in {1,2} {
				\path[dotted,gray] (\x0\z) edge (\x1\z);
			}
		}
		\draw[very thick] (000) edge node[near end] {$\act_1$} (100)
		(100) edge node[near start] {$\act_2$} (110)
		(110) edge node {$\act_3$} (210);
		\draw[dashed,thick] (000) edge node {$\keyact$} (001)
		(001) edge (011)
		(011) edge (111)
		(111) edge node[near end,'] {$\keyact'$} (112)
		(112) edge (212);
		\end{tikzpicture}
		\caption{
			Example of how the trace $\act_1$, $\act_2$, $\act_3$ can be mimicked by introducing additional actions and moving $\act_2$ to the front (dashed trace).
			Transitions that are drawn in parallel have the same label.
		}
		\label{fig:example_mimicking}
	\end{figure}
\end{example}

We remark that Lemma~\ref{lmm:introduce_key_action} also holds if the reduced LSTS is infinite, but finitely branching.
Therefore, it also follows from Lemma~\ref{lmm:introduce_key_action} that the combination of conditions \textbf{\"A1} and \textbf{\"A2} implies \textbf{\"A3} (all from \cite{Valmari1997}) in finitely branching LSTSs.

\subsection{Deterministic LSTSs}
\label{sec:deterministic_lstss}
As already noted in Section~\ref{sec:counter_example}, strong stubborn sets for deterministic systems do not suffer from the inconsistent labelling problem.
The following lemma, which also appeared in~\cite[Lemma 4.2]{Valmari2017b}, shows why.

\begin{lemma}
	\label{lmm:det_lsts_d1'_implied}
	For deterministic LSTSs, conditions \textbf{D1} and \textbf{D2} together imply \textbf{D1'}.
\end{lemma}
\begin{proof}
	Let \TS be a deterministic LSTS, $\pi = s_0 \transition{\act_1} s_1 \transition{\act_2} \dots \transition{\act_n} s_n \transition{\act} s'_n$ a path in \TS and $\redf$ a reduction function that satisfies \textbf{D1} and \textbf{D2}.
	Furthermore, assume that $\act \in \redf(s_0)$ and $\act_1,\dots,\act_n \notin \redf(s_0)$.
	By applying \textbf{D1}, we obtain the path $\pi' = s_0 \transition{\act} s'_0 \transition{\act_1} \dots \transition{\act_n} s'_n$, which satisfies the first part of condition \textbf{D1'}.
	With \textbf{D2}, we have $s_i \transition{a} s^i_i$ for every $1 \leq i \leq n$.
	Then, we can also apply \textbf{D1} to every path $s_0 \transition{\act_1} \dots \transition{\act_i} s_i \transition{a} s^i_i$ to obtain paths $\pi_i = s_0 \transition{\act} s^i_0 \transition{\act_1} s^i_1 \transition{\act_2} \dots \transition{\act_i} s^i_i$.
	Since \TS is deterministic, every path $\pi_i$ must coincide with a prefix of $\pi'$.
	We conclude that $s^i_i = s'_i$ and so the requirement that $s_i \transition{\act} s'_i$ for every $1 \leq i \leq n$ is also satisfied.
	\qed
\end{proof}

\section{Safe Logics}
\label{sec:safe_formalisms}
In this section, we will identify two logics, \viz reachability and CTL$_{-X}$, which are not affected by the inconsistent labelling problem.
This is either due to their limited expressivity or the extra POR conditions that are required.

\subsection{Reachability properties}
Although the counter-example of Section~\ref{sec:counter_example} shows that stutter-trace equivalence is in general not preserved by stubborn sets, some fragments of LTL$_{-X}$ are preserved.
One such class of properties is reachability properties, which are of the shape $\square f$ or $\Diamond f$, where $f$ is a formula not containing temporal operators.

\begin{theorem}
	\label{thm:preserve_reachability}
	Let $\TS$ be an LSTS and $\redf$ a reduction function that satisfies either \textbf{D0}, \textbf{D1}, \textbf{D2}, \textbf{V} and \textbf{L} or \textbf{D1}, \textbf{D2w}, \textbf{V} and \textbf{L}.
	The reduced LSTS induced by $\redf$ is called $\TS_r$.
	For all possible labellings $l \subseteq \AP$, $\TS$ contains a path to a state $s$ such that $L(s) = l$ iff $\TS_r$ contains a path to a state $s'$ such that $L(s') = l$.
\end{theorem}
\begin{proof}
	The `if' case is trivial, since $\TS_r$ is a subgraph of $\TS$.
	For the `only if' case, we reason as follows.
	Let $\TS = (S, \edgerel, \init{s}, L)$ be an LSTS and $\pi = s_0 \transition{\act_1} \dots \transition{\act_n} s_n$ a path such that $s_0 = \init{s}$.
	We mimic this path by repeatedly taking some enabled action $\act$ that is in the stubborn set, according to the following schema.
	Below, we assume the path to be mimicked contains at least one visible action.
	Otherwise, its first state would have the same labelling as $s_n$.
	\begin{enumerate}
		\item If there is an $i$ such that $\act_i \in \redf(s_0)$, we consider the smallest such $i$, \ie, $\act_1,\dots,\act_{i-1} \notin \redf(s_0)$.
		Then, we can shift $\act_i$ forward by \textbf{D1}, move towards $s_n$ along $s_0 \transition{\act_i} s'_0$ and continue by mimicking $s'_0 \transition{\act_1} \dots \transition{\act_{i-1}} s_i \transition{\act_{i+1}} \dots \transition{\act_n} s_n$.
		\item If all of $\act_1,\dots,\act_n \notin \redf(s_0)$, then, by \textbf{D0} and \textbf{D2} or by \textbf{D2w}, we can take a key action $\keyact$ to a state $s'_0$ and continue mimicking the path $s'_0 \transition{\act_1} s'_1 \transition{\act_2} \dots \transition{\act_n} s'_n$.
		Note that $L(s_n) = L(s'_n)$, since $\keyact$ is invisible by condition \textbf{V}.
	\end{enumerate}
	The second case cannot be repeated infinitely often, due to condition \textbf{L}.
	Hence, after a finite number of steps, we reach a state $s'_n$ with $L(s'_n) = L(s_n)$.
	\qed
\end{proof}

We remark that more efficient mechanisms for reachability checking under POR have been proposed, such as condition \textbf{S}~\cite{Valmari2017a}, which can replace \textbf{L}, or conditions based on \emph{up-sets}~\cite{Schmidt2000}.
Another observation is that model checking of safety properties can be reduced to reachability checking by computing the cross-product of a B\"uchi automaton and an LSTS~\cite{BaierKatoen-PMC}, in the process resolving the inconsistent labelling problem.
Peled~\cite{Peled1996} shows how this approach can be combined with POR.

\subsection{Deterministic LSTSs and CTL$_{-X}$ Model Checking}
\label{sec:ctl-x}
In this section, we will consider the inconsistent labelling problem in the setting of CTL$_{-X}$ model checking.
When applying stubborn sets in that context, stronger conditions are required to preserve the branching structure that CTL$_{-X}$ reasons about.
Namely, the original LSTS must be deterministic and one more condition needs to be added~\cite{Gerth1999}:
\begin{description}[leftmargin=!,labelwidth=\widthof{\bfseries C4}]
	\item[\textbf{C4}] Either $\redf(s) = \Act$ or $\redf(s)$ contains precisely one enabled action.
\end{description}
We slightly changed its original formulation to match the setting of stubborn sets.
A weaker condition, called \textbf{\"A8}, which does not require determinism of the whole LSTS is proposed in~\cite{Valmari1997}.
With \textbf{C4}, strong and weak stubborn sets collapse, as shown by the following lemma.

\begin{lemma}
	\label{lmm:c4_implies_d0_d2}
	Conditions \textbf{D2w} and \textbf{C4} together imply \textbf{D0} and \textbf{D2}.
\end{lemma}
\begin{proof}
	Let \TS be an LSTS, $s$ a state and $\redf$ a reduction function that satisfies \textbf{D2w} and \textbf{C4}.
	Condition \textbf{D0} is trivially implied by \textbf{C4}.
	Using \textbf{C4}, we distinguish two cases: either $\redf(s)$ contains precisely one enabled action $\act$, or $\redf(s) = \Act$.
	In the former case, this single action $\act$ must be a key action, according to \textbf{D2w}.
	Hence, \textbf{D2}, which requires that all enabled actions in $\redf(s)$ are key actions, is satisfied.
	Otherwise, if $\redf(s) = \Act$, we consider an arbitrary action $\act$ that satisfies \textbf{D2}'s precondition that $s \transition{\act}$.
	Given a path $s \transition{\act_1 \dots \act_n}$, the condition that $\act_1,\dots,\act_n \notin \redf(s)$ only holds if $n = 0$.
	We conclude that \textbf{D2}'s condition $s \transition{\act_1 \dots \act_n \act}$ is satisfied by the assumption $s \transition{\act}$.
	\qed
\end{proof}

It follows from Lemmata~\ref{lmm:det_lsts_d1'_implied} and \ref{lmm:c4_implies_d0_d2} and Theorem~\ref{thm:correctness_d1'} that CTL$_{-X}$ model checking of deterministic systems with stubborn sets does not suffer from the inconsistent labelling problem.
The same holds for condition \textbf{\"A8}, as already shown in~\cite{Valmari1997}.

\section{Petri Nets}
\label{sec:petri_nets}
Petri nets are a widely-known formalism for modelling concurrent processes and have seen frequent use in the application of stubborn-set theory~\cite{Bonneland2019,Liebke2019,Valmari2017a,Varpaaniemi2005}.
A Petri net contains a set of \emph{places} $P$ and a set of \emph{structural transitions} $T$.
\emph{Arcs} between places and structural transitions are weighted according to a total function $W: (P \times T) \cup (T \times P) \to \mathbb{N}$.
The state space of the underlying LSTS is the set $\Markings$ of all \emph{markings}; a marking $m$ is a function $P \to \mathbb{N}$, which assigns a number of \emph{tokens} to each place.
The LSTS contains a transition $m \transition{\tr} m'$ iff $m(p) \geq W(p,\tr)$ and $m'(p) = m(p) - W(p,\tr) + W(\tr,p)$ for all places $p \in P$.
As before, we assume the LSTS contains some labelling function $L: \Markings \to 2^\AP$.
More details on the labels are given below.
Note that markings and structural transitions take over the role of states and actions respectively.
The set of markings reachable under $\edgerel$ from some \emph{initial marking} $\init{m}$ is denoted $\Mreach$.

\begin{example}
	\label{ex:petri_net}
	Consider the Petri net with initial marking $\init{m}$ below on the left.
	Here, all arcs are weighted 1, except for the arc from $p_5$ to $\tr_2$, which is weighted 2.
	Its LSTS is infinite, but the reachable substructure is depicted on the right.
	The number of tokens in each of the places $p_1,\dots,p_6$ is inscribed in the nodes, the state labels (if any) are written beside the nodes.
	
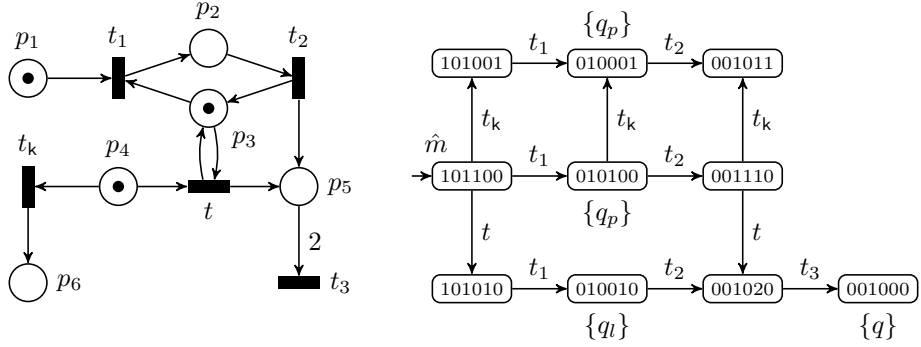
\begin{center}
		\begin{tikzpicture}[->,>=stealth',shorten >=0pt,auto,node distance=2.0cm,semithick,
		every place/.style={draw,minimum size=5mm}]
		
		\begin{scope}[scale=0.8]
		\tikzstyle{vtransition} = [fill,inner sep=0pt,minimum width=1.6mm,minimum height=5.5mm]
		\tikzstyle{htransition} = [fill,inner sep=0pt,minimum width=5.5mm,minimum height=1.6mm]
		
		\def\x{1.5}
		\def\y{2}
		\node[place,label={above:$p_1$},tokens=1]       (p1) at (0,\y) {};
		\node[place,label={above:$p_2$}]                (p2) at (2*\x,1.25*\y) {};
		\node[place,label={below right:$p_3$},tokens=1] (p3) at (2*\x,0.75*\y) {};
		\node[place,label={above:$p_4$},tokens=1]       (p4) at (\x,0.1*\y) {};
		\node[place,label={right:$p_5$}]                (p5) at (3*\x,0.1*\y) {};
		\node[place,label={right:$p_6$}]                (p6) at (0,-0.7*\y) {};
		
		\node[vtransition,label={above:$\tr_1$}]  (t1) at (\x,\y) {};
		\node[vtransition,label={above:$\tr_2$}]  (t2) at (3*\x,\y) {};
		\node[htransition,label={below:$\tr$}]    (t)  at (2*\x,0.1*\y) {};
		\node[htransition,label={right:$\tr_3$}]  (t3) at (3*\x,-0.7*\y) {};
		\node[vtransition,label={above:$\keytr$}] (tk) at (0,0.1*\y) {};
		
		\path
		(p1) edge (t1) (p3) edge (t1) (t1) edge (p2)
		(p2) edge (t2) (t2) edge (p3)
		(p3) edge[bend left=12] (t.50)  (p4) edge (t)  (t.130)  edge[bend left=12] (p3)
		(t) edge (p5) (t2) edge (p5) (p5) edge node {$2$} (t3)
		(p4) edge (tk) (tk) edge (p6);
		\end{scope}
		
		\begin{scope}[xshift=5.9cm,yshift=-1.2cm]
		\tikzstyle{state}=[draw,inner sep=3pt,rectangle,rounded corners=3pt,node font=\scriptsize]
		\def\x{1.8}
		\def\y{1.5}
		
		\node[state,label={above left:$\init{m}$}] (1)  at (0,\y)      {101100};
		\node[state,label={below:$\{q_p\}$}]       (2)  at (\x,\y)     {010100};
		\node[state]                               (3)  at (2*\x,\y)   {001110};
		\node[state]                               (4)  at (0,0)       {101010};
		\node[state,label={below:$\{q_l\}$}]       (5)  at (\x,0)      {010010};
		\node[state]                               (6)  at (2*\x,0)    {001020};
		\node[state]                               (7)  at (0,2*\y)    {101001};
		\node[state,label={above:$\{q_p\}$}]       (8)  at (\x,2*\y)   {010001};
		\node[state]                               (9)  at (2*\x,2*\y) {001011};
		\node[state,label={below:$\{q\}$}]         (10) at (3*\x,0)    {001000};
		\path (-0.8,\y) edge (1)
		(1)  edge    node {$\tr$}   (4)
		(4)  edge    node {$\tr_1$} (5)
		(5)  edge    node {$\tr_2$} (6)
		(1)  edge    node {$\tr_1$} (2)
		(2)  edge    node {$\tr_2$} (3)
		(3)  edge    node {$\tr$}   (6)
		(1)  edge['] node {$\keytr$} (7)
		(7)  edge    node {$\tr_1$} (8)
		(8)  edge    node {$\tr_2$} (9)
		(3)  edge['] node {$\keytr$} (9)
		(2)  edge['] node {$\keytr$} (8)
		(6)  edge    node {$\tr_3$} (10)
		;
		\end{scope}
		\end{tikzpicture}
	\end{center}
	The LSTS practically coincides with the counter-example of Section~\ref{sec:counter_example}.
	Only the self-loops are missing and the state labelling, with atomic propositions $q$, $q_p$ and $q_l$, differs slightly; the latter will be explained later.
	For now, note that $\tr$ and $\keytr$ are invisible and that the trace $\emptyset \{q_p\} \emptyset \emptyset \{q\}$, which occurs when firing transitions $\tr_1 \tr_2 \tr \tr_3$ from $\init{m}$, can be lost when reducing with weak stubborn sets.
	\qed
\end{example}

In the remainder of this section, we fix a Petri net $(P, T, W, \init{m})$ and its LSTS $(\Markings, \edgerel, \init{m}, L)$.
Below, we consider three different types of atomic propositions.
Firstly, polynomial propositions~\cite{Bonneland2019} are of the shape $f(p_1,\dots,p_n) \bowtie k$ where $f$ is a polynomial over $p_1,\dots,p_n$, $\bowtie\, \in \{<,\leq,>,\geq,=,\neq\}$ and $k \in \mathbb{Z}$.
Such a proposition holds in a marking $m$ iff $f(m(p_1),\dots,m(p_n)) \bowtie k$.
A linear proposition~\cite{Liebke2019} is similar, but the function $f$ over places must be linear and $f(0,\dots,0) = 0$, \ie, linear propositions are of the shape $k_1 p_1 + \dots + k_n p_n \bowtie k$, where $k_1,\dots,k_n,k \in \mathbb{Z}$.
Finally, we have arbitrary propositions~\cite{Varpaaniemi2005}, whose shape is not restricted and which can hold in any given set of markings.

Several other types of atomic propositions can be encoded as polynomial propositions.
For example, $\mathit{fireable}(\tr)$~\cite{Bonneland2019,Liebke2019}, which holds in a marking $m$ iff $\tr$ is enabled in $m$, can be encoded as $\prod_{p\in P} \prod_{i = 0}^{W(p,t)-1} (p - i) \geq 1$.
The proposition $\mathit{deadlock}$, which holds in markings where no structural transition is enabled, does not require special treatment in the context of POR, since it is already preserved by \textbf{D1} and \textbf{D2w}.
The sets containing all linear and polynomial propositions are henceforward called $\AP_l$ and $\AP_p$, respectively.
The corresponding labelling functions are defined as $L_l(m) = L(m) \cap \AP_l$ and $L_p(m) = L(m) \cap \AP_p$ for all markings $m$.
Below, the two stutter equivalences $\stuteq_{L_l}$ and $\stuteq_{L_p}$ that follow from the new labelling functions are abbreviated $\stuteq_l$ and $\stuteq_p$, respectively.
Note that $\AP \supseteq \AP_p \supseteq \AP_l$ and $\stuteq \conslab \stuteq_p \conslab \stuteq_l$.

For the purpose of introducing several variants of invisibility, we reformulate and generalise the definition of invisibility from Section~\ref{sec:preliminaries}.
Given an atomic proposition $q \in \AP$, a relation $\strel \subseteq \Markings \times \Markings$ is \emph{$q$-invisible} if and only if $(m, m') \in \strel$ implies $q \in L(m) \Leftrightarrow q \in L(m')$.
We consider a structural transition $\tr$ $q$-invisible iff its corresponding relation $\{(m,m') \mid m \transition{\tr} m' \}$ is $q$-invisible.
Invisibility is also lifted to sets of atomic propositions: given a set $\AP' \subseteq \AP$, relation $\strel$ is \emph{$\AP'$-invisible} iff it is $q$-invisible for all $q \in \AP'$.
If $\strel$ is $\AP$-invisible, we plainly say that $\strel$ is \emph{invisible}.
$\AP'$-invisibility and invisibility carry over to structural transitions.
We sometimes refer to invisibility as \emph{ordinary invisibility} for emphasis.
Note that the set of invisible structural transitions $\Inv$ is no longer an under-approximation, but contains exactly those structural transitions $\tr$ for which $m \transition{\tr} m'$ implies $L(m) = L(m')$ (cf. Section~\ref{sec:preliminaries}).

We are now ready to introduce three orthogonal variations on invisibility.
Firstly, relation $\strel \subseteq \Markings \times \Markings$ is \emph{reach $q$-invisible}~\cite{Valmari2017a} iff $\strel \cap (\Mreach \times \Mreach)$ is $q$-invisible, \ie, all the pairs of reachable markings $(m,m') \in \strel$ agree on the labelling of $q$.
Secondly, $\strel$ is \emph{value $q$-invisible} if 
\begin{enumerate*}[label=\textnormal{(\roman*)}]
	\item $q$ is polynomial and for all $(m,m') \in \strel$, $f(m(p_1),\dots,m(p_n)) = f(m'(p_1),\dots,m'(p_n))$; or if
	\item $q$ is not polynomial and $\strel$ is $q$-invisible.
\end{enumerate*}
Intuitively, this means that the value of polynomial $f$ never changes between two markings $(m,m') \in \strel$.
Reach and value invisibility are lifted to structural transitions and sets of atomic propositions as before, \ie, by taking $\strel = \{ (m,m') \mid m \transition{\tr} m' \}$ when considering invisibility of $\tr$.
\setlength\intextsep{1em}
\begin{wrapfigure}{r}{3.5cm}
	\centering
	\begin{tikzpicture}[->,>=stealth',inner sep=2pt,x={(1cm,0.85cm)},y=0.8cm,z={(-1cm,0.85cm)}]
	\node (is)   at (0,1,1) {$\Inv_s$};
	\node (iv)   at (1,1,0) {$\Inv_v$};
	\node (i)    at (0,1,0) {$\Inv$};
	\node (irs)  at (0,0,1) {$\Inv^r_s$};
	\node (irv)  at (1,0,0) {$\Inv^r_v$};
	\node (ir)   at (0,0,0) {$\Inv^r$};
	\node (irsv) at (1,0,1) {$\Inv^r_{sv}$};
	\node (isv)  at (1,1,1) {$\Inv_{sv}$};
	\path
	(is)  edge (i)
	(iv)  edge (i)
	(irs) edge (ir)
	(irv) edge (ir)
	(is)  edge (irs)
	(iv)  edge (irv)
	(i)   edge (ir)
	(isv) edge (irsv)
	(isv) edge (is)
	(isv) edge (iv)
	(irsv) edge (irs)
	(irsv) edge (irv)
	;
	\end{tikzpicture}
	\caption{Lattice of sets of invisible actions. Arrows represent a subset relation.}
	\label{fig:invisibility_lattice}
\end{wrapfigure}
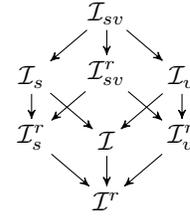
Finally, we introduce another way to lift invisibility to structural transitions: $\tr$ is \emph{strongly $q$-invisible} iff the set $\{ (m,m') \mid \forall p \in P: m'(p) = m(p) + W(t,p) - W(p,t) \}$ is $q$-invisible.
Strong invisibility does not take the presence of a transition $m \transition{\tr} m'$ into account, and purely reasons about the effects of $\tr$.
Value invisibility and strong invisibility are new in the current work, although strong invisibility was inspired by the notion of invisibility that is proposed by Varpaaniemi in~\cite{Varpaaniemi2005}.

We indicate the sets of all value, reach and strongly invisible structural transitions with $\Inv_v$, $\Inv^r$ and $\Inv_s$ respectively.
Since $\Inv_v \subseteq \Inv$, $\Inv_s \subseteq \Inv$ and $\Inv \subseteq \Inv^r$, the set of all their possible combinations forms the lattice shown in Figure~\ref{fig:invisibility_lattice}.
In the remainder, the weak equivalence relations that follow from each of the eight invisibility notions are abbreviated, \eg, $\weakeq_{\Inv^r_{sv}}$ becomes $\weakeq^r_{sv}$.

\begin{example}
	\label{ex:petri_net_AP}
	Consider again the Petri net and LSTS from Example~\ref{ex:petri_net}.
	We can define $q_l$ and $q_p$ as linear and polynomial propositions, respectively:
	\begin{itemize}
		\item $q_l := p_3 + p_4 + p_6 = 0$ is a linear proposition, which holds when neither $p_3$, $p_4$ nor $p_6$ contains a token.
		Structural transition $\tr$ is $q_l$-invisible, because $m \transition{\tr} m'$ implies that $m(p_3) = m'(p_3) \geq 1$, and thus neither $m$ nor $m$ is labelled with $q_l$.
		On the other hand, $\tr$ is not value $q_l$-invisible (by the transition $101100 \transition{\tr} 101010$) or strongly reach $q_l$-invisible (by $010100$ and $010010$).
		However, $\keytr$ is strongly value $q_l$-invisible: it moves a token from $p_4$ to $p_6$ and hence never changes the value of $p_3 + p_4 + p_6$.
		\item $q_p := (1 - p_3)(1 - p_5) = 1$ is a polynomial proposition, which holds in all reachable markings $m$ where $m(p_3) = 0$ and $m(p_5) = 0$.
		Structural transition $\tr$ is reach value $q_p$-invisible, but not $q_p$-invisible (by $002120 \transition{\tr} 002030$) or strongly reach $q_p$ invisible.
		Strong value $q_p$-invisibility of $\keytr$ follows immediately from the fact that the adjacent places of $\keytr$, \viz $p_4$ and $p_6$, do not occur in the definition of $q_p$.
	\end{itemize}
	This yields the state labelling which is shown in Example~\ref{ex:petri_net}.
	\qed
\end{example}

Given a weak equivalence relation $R_\weakeq$ and a stutter equivalence relation $R_\stuteqsym$, we write $R_\weakeq \conslab R_\stuteqsym$ to indicate that $R_\weakeq$ and $R_\stuteqsym$ yield consistent labelling.
We spend the rest of this section investigating under which notions of invisibility and propositions from the literature, the LSTS of a Petri net is labelled consistently.
More formally, we check for each weak equivalence relation $R_\weakeq$ and each stutter equivalence relation $R_\stuteqsym$ whether $R_\weakeq \conslab R_\stuteqsym$.
This tells us when existing stubborn set theory can be applied without problems.
The two lattices containing all weak and stuttering equivalence relations are depicted in Figure~\ref{fig:weak_stut_lattices}; each dotted arrow represents a consistent labelling result.
Before we continue, we first introduce an auxiliary lemma.
\begin{figure}
	\centering
	\begin{tikzpicture}[->,>=stealth',inner sep=2pt]
	\begin{scope}[name prefix={},x={(1cm,0.85cm)},y=0.8cm,z={(-1cm,0.85cm)}]
		\node (ws)   at (0,1,1) {$\weakeq_s$};
		\node (wv)   at (1,1,0) {$\weakeq_v$};
		\node (w)    at (0,1,0) {$\weakeq$};
		\node (wrs)  at (0,0,1) {$\weakeq^r_s$};
		\node (wrv)  at (1,0,0) {$\weakeq^r_v$};
		\node (wr)   at (0,0,0) {$\weakeq^r$};
		\node (wrsv) at (1,0,1) {$\weakeq^r_{sv}$};
		\node (wsv)  at (1,1,1) {$\weakeq_{sv}$};
		\path
		(ws)  edge (w)
		(wv)  edge (w)
		(wrs) edge (wr)
		(wrv) edge (wr)
		(ws)  edge (wrs)
		(wv)  edge (wrv)
		(w)   edge (wr)
		(wsv) edge (wrsv)
		(wsv) edge (ws)
		(wsv) edge (wv)
		(wrsv) edge (wrs)
		(wrsv) edge (wrv)
		;
	\end{scope}
	\begin{scope}[xshift=5cm,yshift=0.2cm,name prefix={}]
		\def\d{1.1}
		\node (s)  at (0,2*\d) {$\stuteq$};
		\node (sp) at (0,\d) {$\stuteq_p$};
		\node (sl) at (0,0) {$\stuteq_l$};
		\path
		(s) edge (sp)
		(sp) edge (sl)
		;
	\end{scope}
	\path[dotted]
	(wrs)  edge[bend left=17] node[near end,fill=white] {Theorem~\ref{thm:petri_net_labelled_consistently_arbitrary_ap}} (s)
	(wv) edge node[fill=white] {Theorem~\ref{thm:petri_net_labelled_consistently_polynomial_ap}} (sp)
	(wrv) edge node[fill=white] {Theorem~\ref{thm:petri_net_labelled_consistently_linear_ap}} (sl)
	;
	\end{tikzpicture}
	\caption{Two lattices containing variations of weak equivalence and stutter equivalence, respectively.
	Solid arrows indicate a subset relation inside the lattice; dotted arrows follow from the indicated theorems and show when the LSTS of a Petri net is labelled consistently.
	}
	\label{fig:weak_stut_lattices}
\end{figure}
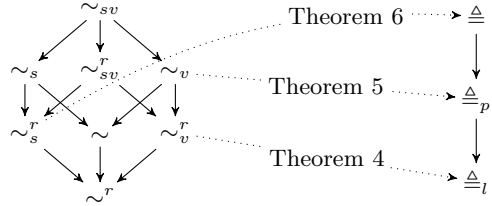

\begin{lemma}
	\label{lmm:meta_proof_pn_consistent_labelling}
	Let $I$ be a set of invisible structural transitions and $L$ some labelling function.
	If for all $\tr \in I$ and paths $\pi = m_0 \transition{\tr_1} m_1 \transition{\tr_2} \dots$ and $\pi' = m_0 \transition{\tr} m'_0 \transition{\tr_1} m'_1 \transition{\tr_2} \dots$, it holds that $\pi \stuteq_L \pi'$, then $\weakeq_I \conslab \stuteq_L$.
\end{lemma}
\begin{proof}
	We assume that the following holds for all paths and $\tr \in I$:
	\begin{equation}
		\tag{$\dagger$}\label{eq:pn_introduce_invis_front}
		m_0 \transition{\tr_1} m_1 \transition{\tr_2} \dots \stuteq_L m_0 \transition{\tr} m'_0 \transition{\tr_1} m'_1 \transition{\tr_2} \dots
	\end{equation}
	We consider two paths $\pi$ and $\pi'$ such that $\pi \weakeq_I \pi'$ and prove that $\pi \stuteq_L \pi'$.
	The proof proceeds by induction on the combined number of invisible structural transitions (taken from $I$) in $\pi$ and $\pi'$.
	In the base case, $\pi$ and $\pi'$ contain only visible structural transitions, and $\pi \weakeq_I \pi$' implies $\pi = \pi'$ since Petri nets are deterministic.
	Hence, $\pi \stuteq_L \pi'$.
	
	For the induction step, we take as hypothesis that, for all paths $\pi$ and $\pi'$ that together contain at most $k$ invisible structural transitions, $\pi \weakeq_I \pi'$ implies $\pi \stuteq_L \pi'$.
	Let $\pi$ and $\pi'$ be two arbitrary paths such that $\pi \weakeq_I \pi'$ and the total number of invisible structural transitions contained in $\pi$ and $\pi'$ is $k$.
	We consider the case where an invisible structural transition is introduced in $\pi'$, the other case is symmetric.
	Let $\pi' = \sigma_1 \sigma_2$ for some $\sigma_1$ and $\sigma_2$.
	Let $\tr \in I$ be some invisible structural transition and $\pi'' = \sigma_1 \tr \sigma'_2$ such that $\sigma_2$ and $\sigma'_2$ contain the same sequence of structural transitions.
	Clearly, we have $\pi' \weakeq_I \pi''$.
	Here, we can apply our original assumption (\ref{eq:pn_introduce_invis_front}), to conclude that $\sigma_2 \stuteq t\sigma'_2$, \ie, the extra stuttering step $\tr$ thus does not affect the labelling of the remainder of $\pi''$.
	Hence, we have $\pi' \stuteq_L \pi''$ and, with the induction hypothesis, $\pi \stuteq_L \pi''$.
	Note that $\pi$ and $\pi''$ together contain $k+1$ invisible structural transitions.
	
	In case $\pi$ and $\pi'$ together contain an infinite number of invisible structural transitions, $\pi \weakeq_I \pi'$ implies $\pi \stuteq_L \pi'$ follows from the fact that the same holds for all finite prefixes of $\pi$ and $\pi'$ that are related by $\weakeq_I$.
	\qed
\end{proof}

The following theorems each focus on a class of atomic propositions and show which notion of invisibility is required for the LSTS of a Petri net to be labelled consistently.
In the proofs, we use a function $d_\tr$, defined as $d_\tr(p) = W(\tr,p) - W(p,\tr)$ for all places $p$, which indicates how structural transition $\tr$ changes the state.
Furthermore, we also consider functions of type $P \to \mathbb{N}$ as vectors of type $\mathbb{N}^{\cardinality{P}}$.
This allows us to compute the pairwise addition of a marking $m$ with $d_\tr$ ($m + d_\tr$) and to indicate that $\tr$ does not change the marking ($d_\tr = 0$).

\begin{theorem}
	\label{thm:petri_net_labelled_consistently_linear_ap}
	Under reach value invisibility, the LSTS underlying a Petri net is labelled consistently for linear propositions, \ie, $\weakeq^r_v \conslab \stuteq_l$.
\end{theorem}
\begin{proof}
	Let $\tr \in \Inv^r_v$ be a reach value invisible structural transition such that there exist reachable markings $m$ and $m'$ with $m \transition{\tr} m'$.
	If such a $\tr$ does not exist, then $\weakeq^r_v$ is the reflexive relation and $\weakeq^r_v \conslab \stuteq_l$ is trivially satisfied.
	Otherwise, let $q := f(p_1,\dots,p_n) \bowtie k$ be a linear proposition.
	Since $\tr$ is reach value invisible and $f$ is linear, we have $f(m) = f(m') = f(m + d_\tr) = f(m) + f(d_\tr)$ and thus $f(d_\tr) = 0$.
	It follows that, given two paths $\pi = m_0 \transition{\tr_1} m_1 \transition{\tr_2} \dots$ and $\pi' = m_0 \transition{\tr} m'_0 \transition{\tr_1} m'_1 \transition{\tr_2} \dots$, the addition of $\tr$ does not influence $f$, since $f(m_i) = f(m_i) + f(d_\tr) = f(m_i + d_\tr) = f(m'_i)$ for all $i$.
	As a consequence, $\tr$ also does not influence $q$.
	With Lemma~\ref{lmm:meta_proof_pn_consistent_labelling}, we deduce that $\weakeq^r_v \conslab \stuteq_l$.
	\qed
\end{proof}

Whereas in the linear case one can easily conclude that $\pi$ and $\pi'$ are stutter equivalent under $f$, in the polynomial case, we need to show that $f$ is constant under all value invisible structural transitions $\tr$, even in markings where $\tr$ is not enabled.
This follows from the following proposition.

\begin{restatable}{proposition}{valueinvisfconstant}
	\label{prop:value_invis_f_constant}
	Let $f: \mathbb{N}^n \to \mathbb{Z}$ be a polynomial function, $a,b \in \mathbb{N}^n$ two constant vectors and $c = a - b$ the difference between $a$ and $b$.
	Assume that for all $x \in \mathbb{N}^n$ such that $x \geq b$, where $\geq$ denotes pointwise comparison, it holds that $f(x) = f(x + c)$.
	Then, $f$ is constant in the vector $c$, \ie, $f(x) = f(x+c)$ for all $x \in \mathbb{N}^n$.
\end{restatable}
\begin{proof}
	Let $f$, $a$, $b$ and $c$ be as above and let $\mathbf{1} \in \mathbb{N}^n$ be the vector containing only ones.
	Given some arbitrary $x \in \mathbb{N}^n$, consider the function $g_x(t) = f(x + t\cdot\mathbf{1} + c) - f(x + t\cdot\mathbf{1})$.
	For sufficiently large $t$, it holds that $x + t \cdot \mathbf{1} \geq b$, and it follows that $g_x(t) = 0$ for all sufficiently large $t$.
	This can only be the case if $g_x$ is the zero polynomial, \ie, $g_x(t) = 0$ for all $t$.
	As a special case, we conclude that $g_x(0) = f(x + c) - f(x) = 0$.
	\qed
\end{proof}

The intuition behind this is that $f(x+c) - f(x)$ behaves like the directional derivative of $f$ with respect to $c$.
If the derivative is equal to zero in infinitely many $x$, $f$ must be constant in the direction of $c$.
We will apply this result in the following theorem.

\begin{theorem}
	\label{thm:petri_net_labelled_consistently_polynomial_ap}
	Under value invisibility, the LSTS underlying a Petri net is labelled consistently for polynomial propositions, \ie, $\weakeq_v \conslab \stuteq_p$.
\end{theorem}
\begin{proof}
	Let $\tr \in \Inv_v$ be a value invisible structural transition, $m$ and $m'$ two markings with $m \transition{\tr} m'$, and $q := f(p_1,\dots,p_n) \bowtie k$ a polynomial proposition.
	Note that infinitely many such (not necessarily reachable) markings exist in $\Markings$, so we can apply Proposition~\ref{prop:value_invis_f_constant} to obtain $f(m) = f(m + d_\tr)$ for all markings $m$.
	It follows that, given two paths $\pi = m_0 \transition{\tr_1} m_1 \transition{\tr_2} \dots$ and $\pi' = m_0 \transition{\tr} m'_0 \transition{\tr_1} m'_1 \transition{\tr_2} \dots$, the addition of $\tr$ does not alter the value of $f$, since $f(m_i) = f(m_i + d_\tr) = f(m'_i)$ for all $i$.
	As a consequence, $\tr$ also does not change the labelling of $q$.
	Application of Lemma~\ref{lmm:meta_proof_pn_consistent_labelling} yields $\weakeq_v \conslab \stuteq_p$.
	\qed
\end{proof}

Varpaaniemi shows that the LSTS of a Petri net is labelled consistently for arbitrary propositions under his notion of invisibility~\cite[Lemma 9]{Varpaaniemi2005}.
Our notion of strong visibility, and especially strong reach invisibility, is weaker than Varpaaniemi's invisibility, so we generalise the result to $\weakeq^r_s \conslab \stuteq$.
\begin{theorem}
	\label{thm:petri_net_labelled_consistently_arbitrary_ap}
	Under strong reach visibility, the LSTS underlying a Petri net is labelled consistently for arbitrary propositions, \ie, $\weakeq^r_s \conslab \stuteq$.
\end{theorem}
\begin{proof}
	Let $\tr \in \Inv^r_s$ be a strongly reach invisible structural transition and $\pi = m_0 \transition{\tr_1} m_1 \transition{\tr_2} \dots$ and $\pi' = m_0 \transition{\tr} m'_0 \transition{\tr_1} m'_1 \transition{\tr_2} \dots$ two paths.
	Since, $m'_i = m_i + d_\tr$ for all $i$, it holds that either
	\begin{enumerate*}[label=\textnormal{(\roman*)}]
		\item $d_\tr = 0$ and $m_i = m'_i$ for all $i$; or
		\item each pair $(m_i,m'_i)$ is contained in $\{ (m,m') \mid \forall p \in P: m'(p) = m(p) + W(t,p) - W(p,t) \}$, which is the set that underlies strong reach invisibility of $\tr$.
	\end{enumerate*}
	In both cases, $L(m_i) = L(m'_i)$ for all $i$.
	It follows from Lemma~\ref{lmm:meta_proof_pn_consistent_labelling} that $\weakeq^r_s \conslab \stuteq$.
	\qed
\end{proof}

To show that the results of the above theorems cannot be strengthened, we provide two negative results.
\begin{theorem}
	\label{thm:petri_net_not_labelled_consistently_weak_inv}
	Under ordinary invisibility, the LSTS underlying a Petri net is not necessarily labelled consistently for arbitrary propositions, \ie, $\weakeq \nconslab \stuteq$.
\end{theorem}
\begin{proof}
	Consider the Petri net from Example~\ref{ex:petri_net} with the arbitrary proposition $q_l$.
	Disregard $q_p$ for the moment.
	Structural transition $\tr$ is $q_l$-invisible, hence the paths corresponding to $\tr_1 \tr_2 \tr \tr_3$ and $\tr \tr_1 \tr_2 \tr_3$ are weakly equivalent under ordinary invisibility.
	However, they are not stutter equivalent.
	\qed
\end{proof}
\begin{theorem}
	\label{thm:petri_net_not_labelled_consistently_reach_value_inv}
	Under reach value invisibility, the LSTS underlying a Petri net is not necessarily labelled consistently for polynomial propositions, \ie, $\weakeq^r_v \nconslab \stuteq_p$.
\end{theorem}
\begin{proof}
	Consider the Petri net from Example~\ref{ex:petri_net} with the polynomial proposition $q_p := (1-p_3)(1 - p_5) = 1$ from Example~\ref{ex:petri_net_AP}.
	Disregard $q_l$ in this reasoning.
	Structural transition $\tr$ is reach value $q_p$-invisible, hence the paths corresponding to $\tr_1 \tr_2 \tr \tr_3$ and $\tr \tr_1 \tr_2 \tr_3$ are weakly equivalent under reach value invisibility.
	However, they are not stutter equivalent for polynomial propositions.
	\qed
\end{proof}

It follows from Theorems~\ref{thm:petri_net_not_labelled_consistently_weak_inv} and~\ref{thm:petri_net_not_labelled_consistently_reach_value_inv} and transitivity of $\subseteq$ that Theorems~\ref{thm:petri_net_labelled_consistently_linear_ap}, \ref{thm:petri_net_labelled_consistently_polynomial_ap} and~\ref{thm:petri_net_labelled_consistently_arbitrary_ap} cannot be strengthened further.
In terms of Figure~\ref{fig:weak_stut_lattices}, this means that the dotted arrows cannot be moved downward in the lattice of weak equivalences and cannot be moved upward in the lattice of stutter equivalences.
The implications of these findings on related work will be discussed in the next section.

\section{Related Work}
\label{sec:related_work}
There are many works in literature that apply stubborn sets.
We will consider several works that aim to preserve LTL$_{-X}$ and discuss whether they are correct when it comes to the problem presented in the current work.

Liebke and Wolf~\cite{Liebke2019} present an approach for efficient CTL model checking on Petri nets.
For some formulas, they can reduce CTL model checking to LTL model checking, which allows greater reductions under POR.
They rely on the incorrect LTL preservation theorem, and since they apply the techniques on Petri nets with ordinary invisibility, their theory is incorrect (Theorem~\ref{thm:petri_net_not_labelled_consistently_weak_inv}).
Similarly, the overview of stubborn set theory presented by Valmari and Hansen in~\cite{Valmari2017a} applies reach invisibility and does not necessarily preserve LTL$_{-X}$.
Varpaaniemi~\cite{Varpaaniemi2005} also applies stubborn sets to Petri nets, but relies on a visibility notion that is stronger than strong invisibility.
The correctness of these results is thus not affected (Theorem~\ref{thm:petri_net_labelled_consistently_arbitrary_ap}).
The approach of B{\o}nneland \etal~\cite{Bonneland2019} operates on two-player Petri nets, but only aims to preserve reachability and consequently does not suffer from the inconsistent labelling problem.

A generic implementation of weak stubborn sets is proposed by Laarman \etal~\cite{Laarman2016}.
They use abstract concepts such as guards and transition groups to implement POR in a way that is agnostic of the input language.
The theory they present includes condition \textbf{D1}, which is too weak.
The implementation relies on a number of binary relations on actions, which they call \emph{accordance} relations.
Based on the accordance relations, an approximation of \textbf{D1} can be computed inductively.
Since the inductive procedure also yields the existence of vertical $\act$ transitions in every $s_i$ (cf. Figure~\ref{fig:condition_D1}), this approximation actually coincides with \textbf{D1'}, and thus their implementation is correct by Theorem~\ref{thm:correctness_d1'}
The implementations proposed in~\cite{Valmari2017a,Wolf2018} are similar, albeit specific for Petri nets.

Others~\cite{Hansen2014,Havelund2015} perform action-based model checking and thus strive to preserve weak trace equivalence or inclusion.
As such, they do not suffer from the problems discussed here, which applies only to state labels.

Although Bene\v{s} \etal~\cite{Benes2011,Benes2009} rely on ample sets, and not on stubborn sets, they also discuss weak trace equivalence and stutter-trace equivalence.
In fact, they present an equivalence relation for traces that is a combination of weak and stutter equivalence.
The paper includes a theorem that weak equivalence implies their new state/event equivalence~\cite[Theorem 6.5]{Benes2011}.
However, the counter-example on the right shows that this consistent labelling theorem does not hold.
\setlength\intextsep{3pt}
\begin{wrapfigure}{r}{3.5cm}
	\centering
	\begin{tikzpicture}[->,>=stealth',shorten >=0pt,auto,node distance=2.0cm,semithick]
		
		\tikzstyle{state} = [draw,circle]
		
		\node[state] (s0)                       at (0,0) {};
		\node[state] (s1)                       at (1.2,0.5) {};
		\node[state,label={above:$\{q\}$}] (s2) at (2.4,0.5) {};
		\node[state] (s3)                       at (1.2,-0.5) {};
		
		\path (-0.5,0) edge (s0)
			(s0) edge    node {$\tau$} (s1)
			(s1) edge    node {$a$}    (s2)
			(s0) edge['] node {$a$}    (s3)
		;
	
	\end{tikzpicture}
\end{wrapfigure}
Here, the action $\tau$ is invisible, and the two paths in this transition system are thus weakly equivalent.
However, they are not stutter equivalent, which is a special case of state/event equivalence.
Although the main POR correctness result~\cite[Corollary 6.6]{Benes2011} builds on the incorrect consistent labelling theorem, its correctness does not appear to be affected.
An alternative proof can be constructed based on Lemma~\ref{lmm:shift_action_forward} and~\ref{lmm:introduce_key_action}.

The current work is not the first to point out mistakes in POR theory.
In~\cite{Siegel2019}, Siegel presents a flaw in an algorithm that combines POR and on-the-fly model checking~\cite{Peled1996}.
In that setting, POR is applied on the product of an LSTS and a B\"uchi automaton.
Let $q$ be a state of the LSTS and $s$ a state of the B\"uchi automaton.
While investigating a transition $(q,s) \transition{\act} (q',s')$, condition \textbf{C3}, which---like condition \textbf{L}---aims to solve the action ignoring problem, incorrectly sets $\redf(q,s') = \enabled(q)$ instead of $\redf(q,s) = \enabled(q)$.

\section{Conclusion}
\label{sec:conclusion}
We discussed the inconsistent labelling problem for preservation of stutter-trace equivalence with stubborn sets.
The issue is relatively easy to repair by strengthening condition \textbf{D1}.
For Petri nets, altering the definition of invisibility can also resolve inconsistent labelling depending on the type of atomic propositions.
The impact on applications presented in related works seems to be limited: the problem is typically mitigated in the implementation, since it is very hard to compute \textbf{D1} exactly.
This is also a possible explanation for why the inconsistent labelling problem has not been noticed for so many years.

Since this is not the first error found in POR theory~\cite{Siegel2019}, a more rigorous approach to proving its correctness, \eg using proof assistants, would provide more confidence.

\bibliographystyle{splncs04}
\bibliography{ref}


\end{document}